\documentclass[11pt,letterpaper]{article}

\usepackage{amsmath,amsthm,amsfonts,amssymb}
\usepackage{algorithm}
\usepackage{algpseudocode}
\usepackage{fullpage}
\usepackage[utf8]{inputenc}
\usepackage[dvipsnames]{xcolor}
\usepackage{xspace,enumerate}
\usepackage[shortlabels]{enumitem}
\usepackage[hypertexnames=false,colorlinks=true,urlcolor=Blue,citecolor=Green,linkcolor=BrickRed]{hyperref}
\usepackage{authblk}
\usepackage{cite}
\usepackage{tikz}

\usetikzlibrary{math}

\DeclareMathOperator{\dist}{dist}
\DeclareMathOperator*{\EX}{\mathbb{E}}

\newcommand{\eps}{\epsilon}

\title{Sublinear Average-Case Shortest Paths\\in Weighted Unit-Disk Graphs}
\date{\vspace{-5ex}}
\author[1]{Adam Karczmarz\thanks{\texttt{a.karczmarz@mimuw.edu.pl}. Supported by ERC Consolidator
Grant 772346 TUgbOAT and by the Foundation for Polish Science (FNP) via the START programme.}}
\author[1]{Jakub Pawlewicz\thanks{\texttt{pan@mimuw.edu.pl}. Supported by ERC Consolidator Grant 772346 TUgbOAT.}}
\author[1]{Piotr Sankowski\thanks{\texttt{sank@mimuw.edu.pl}. Supported by ERC Consolidator Grant 772346 TUgbOAT.}}
\affil[1]{Institute of Informatics, University of Warsaw, Poland}

\def\poly{\operatorname{poly}}
\def\polylog{\operatorname{polylog}}

\theoremstyle{plain}
\newtheorem{theorem}{Theorem}
\newtheorem{lemma}[theorem]{Lemma}
\newtheorem{corollary}[theorem]{Corollary}

\newtheorem{fact}[theorem]{Fact}
\newtheorem{observation}[theorem]{Observation}
\newtheorem{definition}[theorem]{Definition}

\newtheorem{remark}[theorem]{Remark}

\begin{document}

\maketitle

\begin{abstract}
  We consider the problem of computing shortest paths in weighted
  unit-disk graphs in constant dimension $d$. Although the single-source and all-pairs variants
  of this problem are well-studied in the plane case, no non-trivial exact distance
  oracles for unit-disk graphs have been known to date, even
  for $d=2$.

  The classical result of Sedgewick and Vitter [Algorithmica '86] shows
  that for weighted unit-disk graphs in the plane 
  the $A^*$ search has average-case performance
  superior to that of a standard shortest path algorithm, e.g., Dijkstra's algorithm.
  Specifically, if the $n$ corresponding points of a weighted unit-disk graph $G$
  are picked from a unit square uniformly at random, and the connectivity radius is
  $r\in (0,1)$, $A^*$ finds a shortest path in $G$ in $O(n)$ expected time when $r=\Omega(\sqrt{\log n/n})$,
  even though $G$ has $\Theta((nr)^2)$ edges in expectation. In other
  words, the work done by the algorithm is in expectation proportional to the number of vertices
  and not the number of edges.


  In this paper, we break this natural barrier and
  show even stronger sublinear time results.  We propose a new
  heuristic approach to
  computing point-to-point exact shortest paths in
  unit-disk graphs. We analyze the average-case behavior
  of our heuristic using the same random graph model as used by Sedgewick and Vitter
  and prove it superior to $A^*$.
	Specifically, we show that, if we are able to report the set of all $k$
  points of $G$ from an arbitrary rectangular region of the plane
  in $O(k + t(n))$ time, then a shortest path between arbitrary two points of
  such a random graph on the plane can be found
  in $O(1/r^2 + t(n))$ expected time.
  In particular, the state-of-the-art range reporting data structures imply
  a sublinear expected bound for all $r=\Omega(\sqrt{\log n/n})$
  and $O(\sqrt{n})$ expected bound for $r=\Omega(n^{-1/4})$
  after only near-linear preprocessing of the point set.

  Our approach naturally generalizes to higher dimensions $d\geq 3$ and yields sublinear
  expected bounds for all $d=O(1)$ and sufficiently large $r$.


\end{abstract}

\section{Introduction}
Computing shortest paths is certainly one of the most fundamental graph problems
and has numerous theoretical and practical applications.
The two classical variants of the shortest paths problem are \emph{single-source shortest paths} (SSSP)
and \emph{all-pairs shortest path} (APSP).
A common generalization of these variants is the \emph{distance oracle} problem,
where we are allowed to preprocess a given network into a (possibly small) data structure
that is later used to answer arbitrary \emph{point-to-point} shortest paths queries.
Clearly, SSSP and APSP algorithms can be viewed as extreme solutions to the
distance oracle problem: the former can be used without any preprocessing to query a distance in near-linear
time, whereas the latter precomputes the answers
to all the $n^2$ possible queries and thus can answer queries in constant time. Hence,
when constructing distance oracles we seek a tradeoff between these two extremes.
Unfortunately, it is not known how to obtain a non-trivial (with both subquadratic
space and sublinear query time) \emph{exact} distance oracle
for general graphs. Subquadratic space and constant query time oracle is only known
for undirected weighted graphs if approximation factor of at least $3$ is allowed~\cite{oracles}.

Due to this theoretical inefficiency of distance oracles, researchers either focus
on special graph classes, or study approximate approaches.
On one hand, near-optimal (in terms of space and query time) exact distance oracles have
been recently described for planar graphs~\cite{10.1145/3313276.3316316}.
On the other hand, for many important graph classes near-optimal $(1+\eps)$-approximate distance
oracles are known~\cite{10.5555/2027127.2027142,10.1145/1146381.1146411}.

Nevertheless, in practice heuristic approaches are usually preferable -- for an overview
of used techniques see~\cite{Bast2016}. However, the term ``heuristic'' in this domain usually refers to ways of speeding up exact algorithms.
There are some examples of heuristics that have been analyzed theoretically and
proved to yield speedups in meaningful settings, see e.g.,~\cite{10.1145/2985473}.

Perhaps the most well-known heuristic approach to speeding up a shortest path
\linebreak computation is a variant of Dijkstra's algorithm called the \emph{$A^*$ search}~\cite{HartNR68}.
This algorithm incorporates a heuristic function that lower-bounds
the distance to the target and uses it to decide which search paths to follow first.
The algorithm is still guaranteed to find the shortest path
to the target vertex, but the number of vertices explored can be much smaller
compared to the standard Dijkstra's algorithm. For example, if vertices
of the network correspond to points in the plane, the Euclidean distance
to the target is a valid and well-working heuristic function. The natural question
arises when such algorithms perform provably better than in the worst-case.

\subsection{Shortest paths in weighted unit-disk graphs}
The seminal result that answers this question is by Sedgewick and Vitter~\cite{sedgewick1986}, who
studied the performance of $A^*$ search on various random geometric graph models.
Perhaps the most interesting of their results concerns the \emph{weighted unit-disk graphs}.
In a weighted unit-disk graph with \emph{connectivity radius} $r$,
vertices correspond to points on the plane.
An edge between two distinct vertices (points) $u,v$ exists in such a graph
if ${||u-v||_2\leq r}$ and has weight $||u-v||_2$.
This class of geometric graphs has been widely studied from the algorithmic
perspective since such graphs can model e.g., ad-hoc communication networks.
A \emph{random weighted unit-disk graph} $G$, given $n$ and radius $r\in (0,1)$, is
obtained from a set of $n$ random points of a unit square $[0,1]^2$.
Note that such a random $G$ has $\Theta((nr)^2)$ edges in expectation.
However, Sedgewick and Vitter~\cite{sedgewick1986} show that, assuming that the neighbors of each vertex in $G$
are stored explicitly, one can compute a point-to-point shortest path
in $G$ using $A^*$ search in $O(n)$ expected time, i.e., independent of $r$
and sublinear in the size of the edge set of~$G$. In other words, they have given an
exact distance oracle for random weighted unit-disk graphs
that in expectation requires $O((nr)^2)$ space and $O(n)$ query time.

Sedgewick and Vitter's result~\cite{sedgewick1986} can be also interpreted as follows: for weighted
unit-disk graphs ${G=(V,E)}$, just storing the graph explicitly allows
$O(n)$-time queries for an average-case graph~$G$.
Whereas such a query time is sublinear in the graph size, the $\Theta((nr)^2)$
space used might be \emph{superlinear} in the graph's description -- observe
that a weighted unit-disk graph can be described using $O(n)$ space
solely with $n$ point locations and the connectivity radius $r$.
In recent years efficient single-source shortest paths algorithm
for weighted unit-disk graphs have been proposed~\cite{CabelloJ15, KaplanMRSS20, wangxue}, culminating
in the $O(n\log^2{n})$ algorithm of Wang and Xue~\cite{wangxue}.
Note that their worst-case bound is near-optimal and almost matches the bound of~\cite{sedgewick1986}
which holds only on average.
All-pairs shortest paths in weighted unit-disk graphs can be computed
slightly faster that running single-source computations~$n$ times~\cite{ChanS16}.
To the best of our knowledge, no exact distance oracle with non-obvious space and
query bounds for this graph class is known.
On the contrary, a very efficient $(1+\eps)$-approximate distance
oracle with near-optimal space, preprocessing, and query bounds was given by Chan and Skrepetos~\cite{ChanS19a}.

The notion of a weighted unit-disk graph naturally generalizes to three- and higher dimensions:
an edge between two vertices appears if the $d$-dimensional balls of radius $r$
at these points intersect. We are not aware of
any non-trivial results on computing shortest paths in such graphs for $d\geq 3$.

\subsection{Our results}
Observe that all of the above algorithms in order
to answer distance queries require work essentially proportional to the number of vertices
and not the number of edges. In this paper, we break this natural barrier and
show an even stronger sublinear time results.

We propose a natural heuristic approach to computing exact shortest paths in weighted unit-disk graphs.
Following Sedgewick and Vitter, we analyze its average-case query
time by studying its performance on a random $n$-vertex graph with connectivity radius $r$
in the unit square $[0,1]^2$, where $r=\Omega\left(\sqrt{\log (n)/n}\right)$.\footnote{This simplifying assumption
has also been made by Sedgewick and Vitter~\cite{sedgewick1986} and
excludes only very sparse graphs with $m=O(n\log{n})$ from our consideration. Moreover,
it is known that if $r=o(\sqrt{\log(n)/n})$, then the random unit-disk graph
is disconnected with high probability~\cite{825799}.}
In this setting, we prove that after near-linear preprocessing,
the query procedure of our average-case distance oracle
has $O(1/r^2+\sqrt{n})$ expected running time.
Formally, we prove:
\begin{theorem}
  Let $r\in (0,1)$ be such that $r=\Omega\left(\sqrt{\log(n)/n}\right)$. Let $G$ be a
  weighted unit-disk graph with connectivity radius $r$ on a set $P$ of $n$ points 
  picked uniformly at random from the unit square $[0,1]^2$.
  Let $\mathcal{D}$ be a data structure that, after preprocessing $P$ in $O(p(n))$ time,
  supports reporting all $k$ points in $P$
  lying in an arbitrary (not necessarily orthogonal) rectangular subregion
  in $O(k+t(n))$ time.
  Then, there exists an exact distance oracle on $G$ with $O(p(n))$ preprocessing time and
  $O(1/r^2+t(n))$ expected query time.
\end{theorem}

The state-of-the-art range searching data structures~\cite{Chan12}
imply that $t(n)=O(\sqrt{n})$ using $O(n)$ space and $O(n\log{n})$ preprocessing.
Consequently, for $r=\Omega(1/n^{1/4})$ the expected query time
is $O(\sqrt{n})$ and it remains truly sublinear
for all  $r=\Omega(\sqrt{\log(n)/n})$ -- improving the running time of Sedgewick and Vitter
in the full range of parameters $r$ they consider.

The general idea behind our heuristic algorithm for computing a shortest $s-t$ path
is fairly intuitive: we run the single-source shortest paths
algorithm limited to increasingly ``fat'' rectangular subregions of $G$  surrounding
the segment $s-t$.
The subregions of interest are computed using a range reporting data structure
which constitutes the only preprocessed information of our oracle.
Since dynamic variants of such range searching data structures are known~\cite{Matousek92} (with query and space
bound matched up to polylogarithmic factors, and polylogarithmic update bounds),
our heuristic distance oracle can be trivially dynamized as well (see Remark~\ref{dynamic}).

Another advantage of our algorithm is that it easily generalizes to higher dimensions. Using new ideas we prove
that for random weighted unit-disk graphs\footnote{Since a disk is a subset of a plane, in higher dimensions $d>2$, it would be perhaps more appropriate to call
such graphs \emph{weighted unit-ball graphs}. However, anyway, we stick to the well-established term \emph{weighted unit-disk graph}
since our main result concerns the plane case $d=2$.}
in~$[0,1]^d$,
the expected query time is $O(\min(1/r^{2d-1},n)+t_d(n))$,
assuming one can report the points from an arbitrary (not necessarily orthogonally aligned)
$d$-dimensional hyperrectangle in $O(t_d(n)+k)$ time.
It is known~\cite{Chan12} that $t_d(n)=O(n^{1-1/d})$
so this expected time is sublinear in $n$ unless $r=\Omega\left(n^{-\frac{1}{2d-1}}\right)$.
It is worth noting that for $d=2$, the expected query time has a ``better'' dependence, i.e., $O(1/r^d)$, on $r$
than for $d\geq 3$ where the dependence is $O(1/r^{2d-1})$.
This is justified by the fact that whereas single-source shortest paths in weighted
unit-disk graphs for $d=2$ can be computed nearly-optimally~\cite{wangxue}, no
non-trivial algorithm like this is known for $d\geq 3$ and we have to resort
to running the standard Dijkstra's algorithm.

We also show that the $O(t_d(n))=O(n^{1-1/d})$ term can be improved for $d\geq 3$
using a different ad-hoc algorithm for reporting the points
of $G$ in sufficiently thin hyperrectangles surrounding the segment
$st$ that we use.
Namely, we show that we can achieve $O(1/r^{2d-1}+n^{1/d})$ expected
query time for any $d\geq 3$.
This improved algorithm is discussed in Section~\ref{s:improved}
ans is also easily dynamized.

Undoubtedly, the technical difficulty of our result lies in the probabilistic
analysis. We use similar approach to the one used by Sedgewick and Vitter~\cite{sedgewick1986} to bound the probability
that the sought path exist in ellipsoidal grid-like regions called \emph{channels}.
However, in order to avoid looking at all the edges incident to a vertex we need to use a new heuristic that allows
us to consider only edges induced within an rectangular region.

Interestingly, we also identify a shortcoming in their original
analysis for the two-dimensional case and give a more delicate argument inspired
by the techniques from so-called \emph{oriented percolation theory}~(see e.g.,~\cite{durrett}). The original result of
Sedgewick and Vitter~\cite{sedgewick1986} wrongly limited the sets of directed paths
going through the channel grid. Thus the resulting probability that a path exists
was overestimated. The more detailed description of the shortcoming of the original
proof is given in Appendix~\ref{ashortcoming}.

We note that for $d=2$ the graph model considered here has been widely studied in the context
of wireless networks~\cite{825799}. For example, Gupta and Kumar~\cite{760846} studied the connectivity of such
networks, and have shown a critical $r$ above which the graph is connected with high probability.
This result was generalized by Penrose~\cite{1098-2418}
to $k$-connectivity. Our result gives the first known sublinear shortest path routing oracle
for such networks. In a sense, our results call for further work on exact distance oracles
for weighted unit disk graphs. In particular it might suggest that near-linear space and
sublinear query time exact distance oracles in the worst-case exist, as proving such result over
random graphs can be a seen as a proof-of-concept for such a possibility.

\newcommand{\ellipsoid}{\mathrm{BE}}
\newcommand{\bbox}{\mathrm{BB}}
\newcommand{\chlenub}{W_{\mathrm{ub}}}

\section{Preliminaries}

A \emph{weighted unit-disk graph} $G=(V,E)$ with \emph{connectivity radius} $r$ is an undirected geometric graph whose vertices are identified
with some $n$ points in $\mathbb{R}^d$, where $d\geq 2$ is a constant.
The edge set of $G$ contains an edge $\{u,v\}$ for all $u=(u_1,\ldots,u_d)$, $v=(v_1,\ldots,v_d)\in V$
such that $||u-v||_2=\sqrt{\sum_{i=1}^d (u_i-v_i)^2}\leq r$.
For brevity, in the following we omit the subscript and write $||x-y||$ instead of $||x-y||_2$.

For $u,v\in V$, by $\dist_G(u,v)$ we denote the length of a shortest $u\to v$ path
in~$G$.

We consider \emph{exact distance oracles} for weighted unit-disk graphs $G$,
i.e., data structures that preprocess $G$ (ideally into a near-linear space
data structure using near-linear time) and then accept point-to-point
distance queries, i.e., given query vertices $u,v\in V$, compute $\dist_G(u,v)$.
The algorithms we propose can be straightforwardly extended to also report
actual shortest paths within the same asymptotic time bound. Hence, we focus only on computing distances.

In order to perform a meaningful average-case analysis of a distance oracle's
query algorithm
on weighted unit-disk graphs for a given $r$, we need to
limit the space of possible graphs. To this end, following Sedgewick and Vitter~\cite{sedgewick1986},
for $r\in (0,1)$ we limit our attention to graphs with all
$n$ points in $[0,1]^d$.
In order to compute the average running time of a shortest path query,
we would like to compute it over all possible such graphs. 
Equivalently, we study the expected running time of a query
algorithm on a \emph{random graph} $G$, where each of $n$ points
is picked uniformly at random from $[0,1]^d$.
Note that in such a case, each vertex $w$ has $\Theta(nr^d)$ neighbors in expectation:
the probability that another vertex $z$ is connected with $w$ with
an edge equals the probability that $z$ is picked in the $d$-dimensional
ball of radius~$r$ around $v$ which clearly has volume $\Theta(r^d)$.

We also assume $r\geq \left(\frac{\beta\log{n}}{n}\right)^{1/d}$ for a sufficiently large constant $\beta>1$.
Then, the random graph $G$ has $\Omega(n\log{n})$ edges in expectation.
For $d=2$, the bound $r=\Omega\left(\left(\frac{\log{n}}{n}\right)^{1/d}\right)$ has also been assumed by Sedgewick and Vitter~\cite{sedgewick1986},
as it greatly simplifies calculations.
Moreover, for $r=o\left(\left(\frac{\log{n}}{n}\right)^{1/d}\right)$, with high probability $G$ is not connected~\cite{760846}.

\section{The distance oracle}\label{sec:oracle}

\subsection{Preprocessing}
Let the coordinates of the $n$ points of an input weighted unit-disk graph $G$ be given.
In the preprocessing phase, in $O(n\log{n})$ time
we build a simplex range searching data structure
on~$V$~\cite{Chan12}. This data structure requires only linear space
and allows $O(n^{1-1/d}+k)$ worst-case time queries
reporting all of the $k$ input points in an arbitrary hyperrectangle (with sides not necessarily
parallel to the axes) of $\mathbb{R}^d$.

\subsection{Query algorithm}
\label{sec:query_alg}
Suppose the query is to compute $\dist_G(s,t)$ for $s,t\in V$.
Let
\begin{equation*}
  w=||t-s||.
\end{equation*}
Clearly, we have $\dist_G(s,t)\geq w$. 
Moreover, in the following we assume $w>r$, since otherwise we trivially have
$\dist_G(s,t)=w$.

Let us first move and rotate the coordinate system so that the origin is now
in~$s$ and the direction of the first axis is the same as $\overrightarrow{st}$,
thus we have $s=(0,0,\ldots,0)$ and $t=(w,0,\ldots,0)$ in the new coordinate system.

\begin{observation}\label{obs-ellipse}
  Let $W\ge w$ denote an upper bound on $\dist(s,t)$.
  If a $s$--$t$ shortest path in~$G$ contains a vertex $x\in V$ then
  \begin{equation}
    \label{eq:ellipsoid}
    ||x-s||+||x-t||\le W.
  \end{equation}
\end{observation}
Inequality \eqref{eq:ellipsoid} describes a set of points contained in a $d$-dimensional ellipsoid.
The first axis of that ellipsoid has length $W/2$, whereas all other $d-1$ axes have length $R$,
where $R$ satisfies $(w/2)^2 + R^2 = (W/2)^2$. Hence:
\begin{equation*}
  R=\frac12\sqrt{W^2-w^2}.
\end{equation*}

Note that the ellipsoid is contained in a $d$-dimensional \emph{bounding box}
\begin{equation}\label{eq:bb}
  \left[-\frac{W-w}2,\frac{W+w}2\right]\times [-R,R]\times \ldots \times [-R,R]
\end{equation}
with first side length equal to $W$ and the other $d-1$ side lengths equal to $2R$.

We will later pick an unbounded increasing function $\chlenub:\mathbb{Z}_{+}\to \mathbb{R}_{+}$
with values depending on $n,d,r$, with the goal of defining increasingly large bounding boxes, as follows.
\begin{definition}
  For a given integer $i\geq 1$, by $\ellipsoid(i)$ we denote the set
  of points satisfying inequality~\eqref{eq:ellipsoid} for  $W=\chlenub(i)$.
  Similarly, by $\bbox(i)$ we
  denote the bounding box as in formula~\eqref{eq:bb} for $W=\chlenub(i)$.
\end{definition}

\newcommand{\imax}{i_{\text{max}}}

Our entire algorithm will be to run a single-source shortest paths algorithm
on the graphs
\begin{equation*}
  G(i) = (V_i,E_i) = G\cap \bbox(i),
\end{equation*}
subsequently for $i=1,2,\ldots,\imax$ (where $\imax$ is to be set later)
until an $s\to t$
path of length no more than $\chlenub(i)$ is found.
If we are successful with that for some $i$, the found path is returned as the shortest $s\to t$ path.
Otherwise, we simply run Dijkstra's algorithm from $s$ on the entire $G$
and either return the found shortest $s\to t$ path,
or return $\infty$ if no path is found.

\begin{lemma}
  The above algorithm is correct.
\end{lemma}
\begin{proof}
  The algorithm clearly stops. Moreover, the final
  Dijkstra step ensures that an $s\to t$ path is found
  if and only if a $s\to t$ path in $G$ exists.

  To prove correctness suppose that $\dist_G(s,t)<\infty$.
  Let $i^*$ be the first $i$ for which \linebreak ${\dist_{G(i^*)}(s,t)\leq \chlenub(i^*)}$,
  if such $i^*$ exists.
  Since $G(i^*)\subseteq G$, $\dist_G(s,t)\leq \dist_{G(i^*)}(s,t)$ and hence $\dist_G(s,t)\leq \chlenub(i^*)$.
  So, by Observation~\ref{obs-ellipse},
  a path of length $\dist_G(s,t)$ has all its vertices in $\ellipsoid(i^*)\subseteq \bbox(i^*)$.
  This proves $\dist_G(s,t)\geq \dist_{G(i^*)}(s,t)$, so in fact
  $\dist_G(s,t)=\dist_{G(i^*)}(s,t)$.

  If $i^*$ does not exists, we run Dijkstra's algorithm on the entire graph $G$,
  so clearly a shortest $s\to t$ path is returned.
\end{proof}

\newcommand{\tgen}{\text{T}_{\text{gen}}}

Let $\tgen^V(i)$ and $\tgen^E(i)$ be the times required to find
sets $V_i$ and $E_i$, respectively.
Since $V_i$ is defined as a subset of $V$ inside a $d$-dimensional
bounding box $\bbox(i)$, it can be clearly computed using a single query
to the preprocessed range searching data structure.
Hence,
\begin{equation*}
  \tgen^V(i)=O(n^{1-1/d}+|V_i|).
\end{equation*}
Denote by $T_d(i)$ the worst-case running time of step $i$.
The cost $T_d(i)$ might differ depending on the algorithm
that we use to find a shortest path in $G(i)$.
Note that $G(i)$ is a weighted unit-disk graph,
so if $d=2$, and we employ the recent nearly-linear (in the number of vertices), albeit difficult to implement,
algorithm of Wang and Xue~\cite{wangxue}, so we have:\footnote{We use $\log{(|V_i|+2)}$ instead of just
$\log{|V_i|}$ to make sure this term is at least a positive constant.}
\begin{equation}
  \label{eq:T_2}
  T_2(i)=O\left(|V_i|\log^2(|V_i|+2)+\tgen^V(i)\right)=O\left(|V_i|\log^2(|V_i|+2)+\sqrt{n}\right).
\end{equation}

On the other hand, if $d>2$, we need to use the simple-minded Dijkstra's algorithm
to find a shortest path in $G(i)$,
so we have
\begin{equation}
  \label{eq:T_d}
  T_d(i)=O\left(|V_i|\log(|V_i|+2)+|E_i|+\tgen^E(i)\right).
\end{equation}


Let $\bar{P}(i)$ be the probability that we fail
to find a path of length at most $\chlenub(i)$
in the graph $G_i$.
The expected running time of the algorithm is then
\begin{equation}\label{eq:runtime}
  O\left(\sum_{i=1}^{\imax} \bar{P}(i-1)\cdot \EX[T_d(i)]+\bar{P}(\imax)\cdot n^2\right).
\end{equation}

We will prove that by choosing
\begin{equation}\label{eq:imax-asymp}
  \imax=\Theta(nr^d),
\end{equation}
and
\begin{equation}\label{eq:chlenub-asymp}
\chlenub(i)=\Theta\left(w\cdot \sqrt{1+\left(\frac{i}{nr^d}\right)^{\frac{2}{d-1}}}\right)=O(w),
\end{equation}
as described precisely in Section~\ref{sec:explicit},
we can obtain the following key bound.
The proof of this bound is covered in Sections~\ref{s:channels} and~\ref{sec:explicit}.

\begin{lemma}
  \label{lemma:barP}
  For $i=1,\ldots,\imax$, $\bar{P}(i)\le e^{-i}.$
\end{lemma}
We now derive bounds on the expected sizes of sets $V_i$ and $E_i$.
\begin{lemma}
  \label{lemma:Vub}
For $i=1,\ldots,\imax$,
  $\EX[|V_i|]=\Theta\left(
  (w/r)^di
  \right).$
\end{lemma}
\begin{proof}
  Clearly, $\EX[|V_i|]$ equals the volume of $\bbox(i)$ times $n$.
  For $W=\chlenub(i)$ we have
  \begin{equation}\label{eq:r}
  R=\frac{1}{2}\sqrt{W^2-w^2}=\Theta\left(w\cdot \left(\frac{i}{nr^d}\right)^{\frac{1}{d-1}}\right).
  \end{equation}
  Since $\bbox(i)$ has size $W\times 2R\times \ldots \times 2R$, its volume is
  \begin{equation*}
    W\cdot (2R)^{d-1}=\Theta(w)\cdot \Theta(R^{d-1})=\Theta(w)\cdot \Theta\left(\frac{w^{d-1}i}{nr^d}\right)=\Theta\left(\frac{1}{n}\cdot \left(\frac{w}{r}\right)^d\cdot i\right).\qedhere
  \end{equation*}
\end{proof}

In order to analyse the running time we will need the following technical lemma.

\begin{lemma}\label{l:chernoff}
  Let $X$ be a random variable from a binomial distribution with $n$ variables and mean $\EX[X]=\mu=\Omega(1)$.
  Then for any constant integer $\alpha\geq 1$:
  \begin{equation*}
    \EX[X\cdot \log^\alpha (X+2)]=O(\EX[X]\cdot \log^\alpha (\EX[X]+2)))= O(\mu\cdot \log^\alpha (\mu+2))).
  \end{equation*}
\end{lemma}
\begin{proof}
  By using a Chernoff bound $P[X\geq (1+\delta)\mu]\leq e^{\frac{-\delta^2\mu}{2+\delta}}$, and the inequality\linebreak
  $(x+y)^\alpha\leq 2^{\alpha-1}(x^\alpha+y^\alpha)$ that holds for any $x,y>0$ $\alpha\in \mathbb{Z}_+$,
  we get:
  \begin{align*}
    \EX[X\cdot \log^\alpha (X+2)] &= \sum_{j=1}^{n} j \log^\alpha (j+2) \cdot \Pr[X=j]\\
                        &\le \sum_{j=1}^{\lceil n/\mu \rceil} P[(j-1) \mu \le X < j\mu ] \cdot j\log^\alpha (j\mu+2)\\
                        &\le \left(\sum_{j={\lceil 2/\mu\rceil}}^{\infty}P[X\geq (j-1)\mu] \cdot j\mu \log^\alpha (j(\mu+2))\right)+O(\mu)\\
                        &\le \left(\sum_{j=\lceil 2/\mu\rceil-2}^{\infty}P[X\geq (1+j)\mu] \cdot (j+2)\mu \log^\alpha ((j+2)(\mu+2))\right)+O(\mu)\\
                        &=O\left(\sum_{j=1}^{\infty}P[X\geq (1+j)\mu] \cdot j\mu \log^\alpha (j(\mu+2))\right)+O(\mu\log^{\alpha}(\mu+2))\\
                        &=O\left(\sum_{j=1}^{\infty}e^{-\frac{j^2\mu}{2+j}} \cdot j\mu \log^\alpha (j(\mu+2))\right)+O(\mu\log^{\alpha}(\mu+2))\\\
                        &=O\left(\sum_{j=1}^{\infty}e^{-\frac{j^2\mu}{2+j}} \cdot j\mu \cdot 2^{\alpha-1}(\log^\alpha (j)+\log^\alpha(\mu+2))\right)+O(\mu\log^{\alpha}(\mu+2))\\
                        &=O\left(\mu \log^\alpha{(\mu+2)}\cdot \sum_{j=1}^{\infty}e^{-\frac{j^2\mu}{2+j}} \cdot j\log^\alpha (j)\right)+O(\mu\log^{\alpha}(\mu+2))\\
                        &=O\left(\mu \log^\alpha{(\mu+2)}\cdot \sum_{j=1}^{\infty}\left(e^\mu\right)^{-j} \cdot j\log^\alpha (j)\right)+O(\mu\log^{\alpha}(\mu+2))\\
                        &=O(\mu\log^\alpha(\mu+2)).
  \end{align*}
  The final step is justified by $\sum_{j=1}^\infty c^{-j}\cdot \poly(j)=O(1)$ for any $c>1$ and $e^\mu=\Omega(1)$.
\end{proof}

\begin{corollary}\label{cor:logs}
  For any integer $\alpha\geq 1$ we have
  \begin{equation*}
    \EX[|V_i|\log^{\alpha}(|V_i|+2)] =O(\EX[|V_i|]\cdot \log^{\alpha}(\EX[|V_i|]+2)).
  \end{equation*}
\end{corollary}
\begin{proof}
  We can apply Lemma~\ref{l:chernoff} since $\EX[|V_i|]=\Omega(1)$ by Lemma~\ref{lemma:Vub}.
\end{proof}

\newcommand{\efun}{f^\text{E}}

\begin{lemma}
  \label{lemma:Eub}
  Let
    \begin{equation*}
      \efun(i)=\min\{nr^d, (w/r)^{d-1}i\}=\begin{cases}
      (w/r)^{d-1}i\quad&\text{for }r\ge \sqrt[2d-1]{w^{d-1}i/n}\\
      nr^d&\text{otherwise.}
    \end{cases}
    \end{equation*}
  Then for $i=1,\ldots,\imax$, $\EX[|E_i|]=\EX[|V_i|]\cdot O(\efun(i))$.

\end{lemma}
\begin{proof}
  Take a vertex $v\in V_i$.
  All neighbours of $v$ in $G(i)$ belong to the intersection of the $d$-dimensional ball of
  radius $r$ centered at $v$, and the bounding box $\bbox(i)$.
  This intersection, on one hand,
  is contained in a box of size $2r\times2R\times\cdots\times 2R$,
  where $R=\frac{1}{2}\sqrt{(\chlenub(i))^2-w^2}$ (see \eqref{eq:r}).
  On the other hand, it is trivially inside a ball of radius $r$.
  In the former case the volume of the box with $v$'s neighbours is
  \begin{equation*}
    O(rR^{d-1})=O\left(\frac{1}{n}(w/r)^{d-1}\cdot i\right)
  \end{equation*}  In the latter case the volume is $O(r^d)$.
  Therefore, the expected number of neighbours of $v$ is
  \begin{equation}
    \label{eq:nb_v}
    O\left(n\cdot \min\left\{\frac{1}{n}(w/r)^{d-1}\cdot i, r^d\right\}\right)=O(\min\{(w/r)^{d-1}\cdot i, nr^d\}).
  \end{equation}
  By linearity of expectation we get the desired bound on $\EX[|E_i|]$.
\end{proof}

The following lemma describes how to efficiently generate the edges $E_i$
when we use Dijkstra's algorithm (for $d\geq 3$).
\begin{lemma}
  \label{lemma:tgenE}
  Let $\efun$ be as in Lemma~\ref{lemma:Eub}.
  Given $V_i$,
  the edge set $E_i$ can be computed in
  $\tgen^E(i)=O(\EX[|V_i|]\cdot \efun(i))$ expected time.
\end{lemma}
\begin{proof}
  We divide $[0,1]^d$ into cubes of size $r\times r\times\cdots\times r$.
  With each non-empty cube we will keep a list of vertices from $V_i$
  that belongs to that cube.
  We build these lists by iterating over all $v\in V_i$ and assign $v$
  to the appropriate cube's list.
  Technically speaking, the lists are stored in a
  hash table with expected $O(1)$ insertion and access time (see e.g.,~\cite{DietzfelbingerH90}):
  note that the cubes can be mapped to integers $[1,(\lceil 1/r\rceil)^d]$
  and we have $(\lceil 1/r\rceil)^d=O(n)$ by $r=\Omega\left((\log(n)/n)^{1/d}\right)$.
  To find the edges,
  for each $v$ we iterate over all vertices $w$ belonging
  to the same cube as $v$ or a neighbouring cube
  and check whether $||v-w||\le r$.
  There are at most $3^d$ such cubes and each neighbor
  of $v$ necessarily lies in these neighboring cubes.

  Each cube contains $O(n\min\{rR^{d-1}, r^d\})$ vertices in expectation,
  where we again set $R=\frac{1}{2}\sqrt{(\chlenub(i))^2-w^2}$ (see \eqref{eq:r}).
  Recall from~\eqref{eq:nb_v} in Lemma \ref{lemma:Eub} that
  this quantity is $O(\efun(i))$.
  This is because if $2R<r$ then the cube's intersection with $\bbox(i)$
  has size at most $r\times(2R)\times\cdots\times(2R)$ and
  only in that part of the cube the vertices from $V_i$ can appear.
  Therefore, the expected total work for each vertex will be $O(3^d\cdot \efun(i))=O(\efun(i))$.
  Thus, by linearity of expectation, the expected running time is indeed $O(\EX[|V_i|]\cdot f^E(i))$.
\end{proof}

We are now ready to prove the following theorem bounding the expected
running time of the query algorithm.
\begin{theorem}\label{t:query-general}
  The expected running time of the query algorithm on
  an $n$-vertex random weighted unit-disk graph in $[0,1]^d$
  with connectivity radius $r$ is
  \begin{enumerate}[(a)]
    \item\label{item:d2}
      $O\left((w/r)^2\log^2(1+w/r)+\sqrt{n}\right)$ for $d=2$,
    \item\label{item:big_r}
      $O\bigl((w/r)^{2d-1}+n^{1-1/d}\bigr)$ for $d\ge3$ and $r\ge\sqrt[2d-1]{w^{d-1}/n}$,
    \item\label{item:small_r}
      $O(nw^d+n^{1-1/d})$ otherwise.
  \end{enumerate}
\end{theorem}
\begin{proof}
  In all cases we will bound the expected query time as given in sum~(\ref{eq:runtime}):
\begin{equation*}
  O\left(\sum_{i=1}^{\imax} \bar{P}(i-1)\cdot \EX[T_d(i)]+\bar{P}(\imax)\cdot n^2\right).
\end{equation*}

  First of all, note that by Equation~\ref{eq:imax-asymp}, Lemma~\ref{lemma:barP} and the assumption
  $r\geq (\beta \log(n)/n)^{1/d}$ where $\beta>1$ is a large enough constant, for some constant $\gamma>0$ we have:
  \begin{equation*}
    \imax\geq \gamma\cdot nr^d\geq \gamma\cdot\beta \log{n}
  \end{equation*}
  So, picking $\beta=2/\gamma$ gives us
  \begin{equation*}
    \bar{P}(\imax)\cdot n^2=O\left(e^{-\imax}\cdot n^2\right)=O\left(e^{-2\log{n}}\cdot n^2\right)=O(1).
  \end{equation*}
  Hence, we can focus on the below sum. By Lemma~\ref{lemma:barP}, we have:
  \begin{equation*}
    O\left(\sum_{i=1}^{\imax} \bar{P}(i-1)\cdot \EX[T_d(i)]\right)
    =O\left(\sum_{i=1}^\infty \EX[T_d(i)]e^{-(i-1)}\right)
    =O\left(\sum_{i=1}^\infty \EX[T_d(i)]e^{-i}\right).
  \end{equation*}
  In the following, we will use the asymptotic formula
  $\sum_{i=1}^\infty f(i)e^{-i}=O(1)$ that holds for any function $f(i)=\poly(i)$.
  Recall that $w>r$.

  Let us first prove item~(\ref{item:d2}). By \eqref{eq:T_2} and Lemma~\ref{lemma:Vub},
  we have:
  \begin{multline*}
    O\left(\sum_{i=1}^\infty \EX[T_2(i)]e^{-i}\right)\\
    =O\left(\sum_{i=1}^\infty (w/r)^2\cdot i\cdot \log^2\bigl((w/r)^2i+2\bigr)\cdot e^{-i}
    +\sum_{i=1}^\infty\sqrt{n}e^{-i}\right)\\
    =O\left(
      (w/r)^2\log^2(w/r+1)\sum_{i=1}^\infty i\log^2(i)\cdot e^{-i}
      +\sqrt{n}\sum_{i=1}^\infty e^{-i}
    \right)\\
    =O\left((w/r)^2\log^2(1+w/r)+\sqrt{n}\right).
  \end{multline*}
  Above we silently used Corollary~\ref{cor:logs} for $X=|V_i|$ and $\alpha=2$.
  Now let us prove items (\ref{item:big_r})~and~(\ref{item:small_r}).
  Let us first argue that the term $\EX[|V_i|\log{|V_i|}]$ is, by Corollary~\ref{cor:logs}, asymptotically
  dominated by the bound $\EX[|V_i|]\cdot O(\efun(i))$ on $\EX[|E_i|]$ from Lemma~\ref{lemma:Eub}.
  This follows by Lemmas \ref{lemma:Vub} and \ref{lemma:Eub} --
  if $r$ is sufficiently large.
  Thus by plugging that bound into \eqref{eq:T_d} we get
  \begin{multline*}
    O\left(\sum_{i=1}^\infty \EX[T_d(i)]e^{-i}\right)\\
    =\sum_{i=1}^\infty(w/r)^di\cdot
    \min\{nr^d, (w/r)^{d-1}i\}
    e^{-i}
    +\sum_{i=1}^\infty n^{1-1/d}e^{-i}
    \\
    =O\left(
    \min\left\{
    nw^d
    \sum_{i\ge 1} ie^{-i}
    ,
    (w/r)^{2d-1}
    \sum_{i=1}^\infty i^2e^{-i}
    \right\}
    +
    n^{1-1/d}
    \sum_{i=1}^\infty
    e^{-i}
    \right)
    \\
    =O\left(
    \min\left\{
    nw^d
    ,
    (w/r)^{2d-1}
    \right\}
    +
    n^{1-1/d}
    \right).\qedhere
  \end{multline*}
\end{proof}

\begin{remark}\label{dynamic}
  The described distance oracle can be very easily made dynamic with
  only polylogarithmic overhead. That is, we can support insertions
  and deletions of vertices of the weighted unit-disk graph $G$,
  in amortized $O(\polylog{n})$ time. To this end we simply
  replace the simplex range query data structure of Chan~\cite{Chan12}
  that we build in the preprocessing with that of Matousek~\cite{Matousek92}
  which allows for polylogarithmic amortized updates to the point set and
  has only polylogarithmically slower preprocessing and query times.
\end{remark}

\subsection{Faster generation of sets $V_i$}\label{s:improved}
We now show how to generate $V_i$ faster and 
without resorting to using simplex range query data
structure~\cite{Chan12}.
Let $k$ be an integer to be chosen later. 
Let us partition $[0,1]^d$ into
$k^d$ orthogonal \emph{cells}, each of size $(1/k)\times (1/k)\times\ldots\times (1/k)$.
For any $(i_1,\ldots,i_d)\in \{1,\ldots,k\}^d$,
the cell $C_{i_1,\ldots,i_d}$ equals $[(i_1-1)\cdot (1/k),i_1\cdot (1/k)]\times \ldots\times [(i_d-1)\cdot (1/k),i_d\cdot (1/k)]$.

During preprocessing, each point of $v\in V$ is assigned to an arbitrary cell $C_v$ out of $O(1)$ cells $v$ is contained in.
Clearly, for any cell $C$ we have $E[V\cap C]=n\cdot (1/k)^d$.

Upon query, each required $V_i=\bbox(i)\cap V$ is generated as follows.
We first find all the cells $\mathcal{C}_i$ that intersect $\bbox(i)$.
To this end, we start by adding the cell $C_s$ to $\mathcal{C}_i$.
For each added $C\in \mathcal{C}_i$ we iterate through its at most $3^d=O(1)$
neighboring cells and add them to $\mathcal{C}_i$ unless they do not intersect $\bbox(i)$.
Since the cells intersecting $\bbox(i)$ form a connected subset of all cells,
this algorithm is correct.
As each cell has $O(1)$ neighbors,
the time used to construct $\mathcal{C}_i$ is linear in the final size
of $\mathcal{C}_i$.
Finally, for each $C\in\mathcal{C}_i$, we iterate through the vertices $v$
assigned to $C$ (i.e., with $C_v=C$) and include $v$ in $V_i$ if $v\in \bbox(i)$.

Clearly, the expected running time of the above algorithm
is $O(|\mathcal{C}_i|\cdot (n\cdot (1/k)^d+1))$.
To proceed, we need to bound the size of $\mathcal{C}_i$.
Recall that $\bbox(i)$ has size $\chlenub(i)\times 2R\times\ldots\times 2R$,
where $R=\frac{1}{2}\sqrt{\chlenub(i)-w^2}=\Theta(w(i/nr^d)^{\frac{1}{d-1}})$.
\begin{lemma}
  We have $|\mathcal{C}_i|=O(k^dR^{d-1}+k)$.
\end{lemma}
\begin{proof}
  Let us partition $\bbox(i)$ into $O(k)$ chunks of size
  $(1/k)\times 2R\times\ldots\times 2R$.
  Each such chunk is contained in a union of 
  $O(\max(Rk,1)^{d-1})$
  hybercubes of size $(1/k)\times \ldots\times (1/k)$.
  The longest diagonal of such a hybercube has length $\sqrt{d}/k$.
  As a result, such a hypercube lies in an orthogonally aligned hypercube of
  size $(\sqrt{d}/k)\times\ldots\times (\sqrt{d}/k)$.
  Clearly, such an aligned hypercube can be covered by $O(d^{d/2})=O(1)$ cells.
  As a result, a chunk can be covered using $O(\max(Rk,1)^{d-1})$ cells.
  Finally, we conclude that $\bbox(i)$ can be covered using
  $O(k\cdot \max(Rk,1)^{d-1})=O(k^dR^{d-1}+k)$ cells.
\end{proof}

By the above lemma, the expected running time is:
\begin{equation*}
  O\left(|\mathcal{C}_i|\cdot \left(n\frac{1}{k^d}+1\right)\right)=O\left((k^dR^{d-1}+k)\cdot \left(\frac{n}{k^d}+1\right)\right)=O\left(i\cdot(1/r^d+k^d/nr^d)+\frac{n}{k^{d-1}}+k\right).
\end{equation*}

By picking $k=\lfloor n^{1/d}\rfloor$, we obtain $O(i/r^d+n^{1/d})$ expected running time.
As a result, through all $i=1,\ldots,\imax$, the expected total time
required to construct the sets $|V_i|$ is
\begin{equation*}
  O\left(\sum_{i=1}^{\imax}\bar{P}(i-1)\cdot (i/r^d+n^{1/d})\right)=O\left((1/r^d+n^{1/d})\sum_{i=1}^{\imax}e^{-i}\cdot i\right)=O(1/r^d+n^{1/d}).
\end{equation*}

By combining the above with our earlier developments, we obtain
the following improved version of Theorem~\ref{t:query-general}.
\begin{theorem}\label{t:query-general-faster}
  The expected running time of the query algorithm on
  an $n$-vertex random weighted unit-disk graph in $[0,1]^d$
  with connectivity radius $r$ is
  \begin{enumerate}[(a)]
    \item\label{item:d2}
      $O\left((w/r)^2\log^2(1+w/r)+\sqrt{n}\right)$ for $d=2$,
    \item\label{item:big_r}
      $O\bigl((w/r)^{2d-1}+n^{1/d}\bigr)$ for $d\ge3$ and $r\ge\sqrt[2d-1]{w^{d-1}/n}$,
    \item\label{item:small_r}
      $O(nw^d+n^{1/d})$ otherwise.
  \end{enumerate}
\end{theorem}

Finally, similarly as in Remark~\ref{dynamic}, we note that the the algorithm in this section
is also efficient if $G$ undergoes dynamic updates, such as point insertions and deletions.
The only data structure we use is the assignment of points to $\Theta(n)$ cubes.
Clearly, this assignment can be easily updated in constant time upon
an insertion or deletion of vertices.


\section{Channels}\label{s:channels}
The remaining part of the paper is devoted to proving the very convenient bound on $\bar{P}(i)$
from Lemma~\ref{lemma:barP}.

We start by introducing a notion of a \emph{channel}, which
is a parameterized grid-like object whose goal is to ``discretize'' the space
of possible shortest $s\to t$ paths in $\bbox(i)$.
The next step is to upper-bound the probability that we fail
to find reasonably short $s\to t$ path in the channel.
Afterwards, we are ready to give explicit formulas
for $\imax$ and $\chlenub(i)$ so that the asymptotic bounds~(\ref{eq:imax-asymp})~and~(\ref{eq:chlenub-asymp}),
as well as the bound $\bar{P}(i)\leq e^{-i}$
hold.

\newcommand{\channel}{\textrm{ch}}

Roughly speaking, a channel is a subset of vertices $V$ restricted to some subspace.
We generalize the channels defined in \cite[page 41]{sedgewick1986} to
$d$-dimensional space and arbitrary start/end vertices $s$ and $t$.

Recall that $w=||t-s||$ and $w>r$. Let $K\geq 1$ be the smallest integer such that $l=w/(4K+1) \le r/4$.
We also have
\begin{equation}\label{eq:ldef}
  l=\frac{w}{4(K-1)+1}\cdot\frac{4(K-1)+1}{4K+1}>r/4\cdot \frac{4K-3}{4K+1}\geq r/20.
\end{equation}
We are going to work in the coordinate system introduced in Section \ref{sec:query_alg}.
Let us denote the first axis by $x_0$ and the remaining axes by $x_1,\ldots,x_{d-1}$.

\begin{definition}[Box $R(z_0,z_1,\ldots,z_{d-1})$]\label{def:box}
  Let $h>0$ be fixed.
  Let us cut the space using planes $x_0=lz$ and $x_i=(1/2+z)h$
  for all integers $z$ and $i=1,\ldots,d-1$.

  For $z_0,z_1,\ldots,z_{d-1}\in\mathbb{Z}$, \emph{the box} $R(z_0,z_1,\ldots,z_{d-1})$ contains all points
  $(x_i)_{i=0}^{d-1}$ satisfying:
  \begin{itemize}
    \item $lz_0\le x_0 \le l(z_0+1)$,
    \item $(-1/2+z_i)h\le x_i\le (1/2+z_i)h$ for all $i=1,\ldots,d-1$.
  \end{itemize}
\end{definition}

Each box, defined as above, has size $l\times h\times\cdots\times h$.
Note that ${s\in R(0,0,\ldots,0)}$ and $t\in R(4K,0,\ldots,0)$.
Now suppose we want to travel from the box containing $s$ to
the box containing $t$
using \emph{jumps}, defined below.

\begin{definition}[Jumping between boxes]\label{def:jump}
  We say that we can \emph{jump} from box
  $R(z_0,z_1,\ldots,z_{d-1})$ to
  box $R(z'_0,z'_1,\ldots,z'_{d-1})$ iff
  \begin{itemize}
    \item $z'_0=z_0+2$,
    \item $|z'_i-z_i|=1$ for all $i=1,\ldots,d-1$.
  \end{itemize}
\end{definition}

Consider a jumping trip from $R(0,0,\ldots,0)$ to $R(4K,0,\ldots,0)$.
\begin{observation}[Reachable boxes]
  \label{obs:reachable_boxes}
  Let $B=R(z_0,z_1,\ldots,z_{d-1})$ be an arbitrary box.
  Suppose a sequence of jumps (as defined above) from $R(0,0,\dots,0)$ to $R(4K,0,\ldots,0)$
  goes through the box $B$. Then, the following conditions hold:
  \begin{itemize}
    \item $z_0=2k$ for some integer $k$, $0\le k\le 2K$,
    \item $|z_i|\le\min(k,2K-k)$ for all $i=1,\ldots,d-1$,
    \item $z_i\equiv k\pmod 2$.
  \end{itemize}
\end{observation}

Now we are ready to define the \emph{channel} parameterized by $h$.
\begin{definition}[Channel]\label{d:channel}
  A \emph{channel}  $\channel(h)$
  is a subset of $[0,1]^d$
  defined as the union of all boxes $B$ satisfying the
  conditions of Observation \ref{obs:reachable_boxes}.
\end{definition}

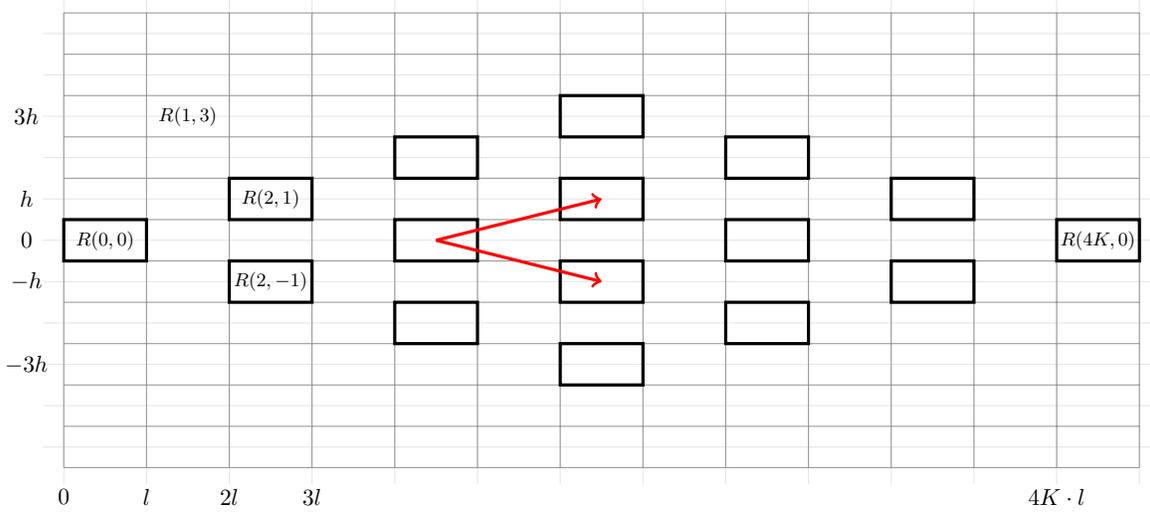
\begin{figure}
  \centering
  \begin{tikzpicture}[scale=0.55, every node/.style={scale=0.80}]

    \draw[xstep=2.0,ystep=1.0,color=black!10] (-0.5,-5.9) grid (26.5,5.9);
    \draw[xstep=2.0,ystep=1.0,color=black!40,shift={(0,0.5)}] (0,-6) grid (26,5);
    
    \draw[very thick] (0,-0.5) rectangle (2,0.5);
    
    \draw[very thick] (4,-1.5) rectangle (6,-0.5);
    \draw[very thick] (4,0.5) rectangle (6,1.5);

    \draw[very thick] (8,-2.5) rectangle (10,-1.5);
    \draw[very thick] (8,-0.5) rectangle (10,0.5);
    \draw[very thick] (8,1.5) rectangle (10,2.5);

    \draw[very thick] (12,-3.5) rectangle (14,-2.5);
    \draw[very thick] (12,-1.5) rectangle (14,-0.5);
    \draw[very thick] (12,0.5) rectangle (14,1.5);
    \draw[very thick] (12,2.5) rectangle (14,3.5);
    
    \draw[very thick] (16,-2.5) rectangle (18,-1.5);
    \draw[very thick] (16,-0.5) rectangle (18,0.5);
    \draw[very thick] (16,1.5) rectangle (18,2.5);

    \draw[very thick] (20,-1.5) rectangle (22,-0.5);
    \draw[very thick] (20,0.5) rectangle (22,1.5);

    \draw[very thick] (24,-0.5) rectangle (26,0.5);
  
    \node at (0,-6.2) {$0$};
    \node at (2,-6.2) {$l$};
    \node at (4,-6.2) {$2l$};
    \node at (6,-6.2) {$3l$};

    \node at (24,-6.2) {$4K\cdot l$};

    \node at (-0.9,0) {$0$};
    \node at (-0.9,1) {$h$};
    \node at (-0.9,-1) {$-h$};
    \node at (-0.9,3) {$3h$};
    \node at (-0.9,-3) {$-3h$};
  
    \node at (1,0) {\footnotesize $R(0,0)$};
    \node at (5,-1) {\footnotesize $R(2,-1)$};
    \node at (5,1) {\footnotesize $R(2,1)$};

    \node at (25,0) {\footnotesize $R(4K,0)$};
    \node at (3,3) {\footnotesize $R(1,3)$};

    \draw[->,very thick,color=red] (9,0) to (13,1);
    \draw[->,very thick,color=red] (9,0) to (13,-1);
  \end{tikzpicture}
  \caption{The rectangles represent boxes from Definition~\ref{def:box} for $d=2$. The red arrows represent possible jumps
  from a single box. The channel $\channel(h)$ for $K=3$ (see Definition~\ref{d:channel}) is represented by rectangles with
  thick black border.}\label{f:boxes}
\end{figure}

In other words, a channel $\channel(h)$ consists of all boxes that
can appear in a sequence of jumps from the box containing
$s$ to the box containing $t$.
Boxes, jumps, and channels are depicted in Figure~\ref{f:boxes}.

In the following, we say that a box $B$ is \emph{empty} if it does not contain any vertex of $G$.
\subsection{Paths in a channel}
Not all channels $\channel(h)$ are of our interest.
We need a condition on $h$ guaranteeing
that if we can jump from a non-empty box $B$ to another non-empty box $B'$
then there exists an appropriate edge in the graph,
namely if there is $u\in B\cap V$ and $v\in B'\cap V$ then $||u-v||\le r$.
Then, a sequence of jumps between non-empty boxes will certify the
existence of a path in $G$.

Observe that the distance between two opposite corners of $B$ and $B'$
(recall that $B$ and $B'$ have to satisfy Definition~\ref{def:jump})
is
\begin{equation*}\sqrt{(3l)^2+(d-1)(2h)^2}.\end{equation*}
We need this to be smaller than $r$.
Taking into account that $l\le r/4$, it is sufficient that
\begin{equation*}
  \Bigl(\frac34r\Bigr)^2+(d-1)(2h)^2\le r^2,
\end{equation*}
which gives
\begin{equation}
  \label{eq:h_bound}
  h\le\frac18\sqrt{\frac7{d-1}}\cdot r.
\end{equation}

\begin{definition}[Path in $\channel(h)$]\label{def:path}
  \emph{A path} in $\channel(h)$
  with $h$ satisfying \eqref{eq:h_bound}
  is a sequence of non-empty boxes $B_0,\ldots,B_{2K}$
  such that $B_0=R(0,0,\ldots,0)$, $B_{2K}=R(4K,0,\ldots,0)$, and
  we can jump from $B_j$ to $B_{j+1}$ for all $j=0,\ldots,{2K-1}$.
\end{definition}

Now we show that a path in $\channel(h)$ certifies the existence of an $s-t$ path in~$G$
which is not too long.
Specifically, we show the following bound.
\begin{lemma}[Channel induced path length]
  \label{lemma:channel_len}
  Suppose there is a path in $\channel(h)$. Then, there
  exists an
  $s-t$ path in $G$ of length no more than
  \begin{equation}
    \label{eq:channel_len}
    w\sqrt{1+40^2(d-1)(h/r)^2}.
  \end{equation}
\end{lemma}

\begin{proof}
  Let $u_j=(u_0^j,\ldots,u_{d-1}^j)$ be a vertex of $G$ in $B_j\cap V$.
  Additionally, set $u_0=s$ and $u_{2K}=t$.
  Recall that $u_j$ exists since each box in a path in $\channel(h)$ is non-empty.
  Consider subsequent vertices $u_j$ and $u_{j+1}$.
  Note that
  \begin{align*}
    ||u_{j+1}-u_j||=\sqrt{\sum_{i=0}^{d-1} (u_i^{j+1}-u_i^j)^2}=(u_0^{j+1}-u_0^j)\sqrt{1+\sum_{i=1}^{d-1}\left(\frac{u_i^{j+1}-u_i^j}{u_0^{j+1}-u_0^j}\right)^2}.
  \end{align*}
  Recall that we have $u_0^{j+1}-u_0^j\geq l$ and $u_i^{j+1}-u_i^j\leq 2h$ for $i\geq 1$. Hence,
  \begin{equation*}||u_{j+1}-u_j||\leq (u_0^{j+1}-u_0^j)\sqrt{1+(d-1)\frac{(2h)^2}{l^2}}.\end{equation*}
  Since $u_0\to u_1\to\ldots u_{2K}$ is a path in $G$, the
  length of a shortest $s-t$ path in $G$ can be bounded by:
  \begin{align*}
    \sum_{j=0}^{2K-1}||u_{j+1}-u_j|| &\leq
    \sqrt{1+(d-1)\frac{(2h)^2}{l^2}}\cdot\sum_{j=0}^{2K-1}(u_0^{j+1}-u_0^j)\\
    &=\sqrt{1+(d-1)\left(\frac{2h}{l}\right)^2}\cdot w.
  \end{align*}
  The claimed bound is obtained by $l\geq r/20$.
\end{proof}


\newcommand{\AntiP}{\hat P}

\subsection{Probability}

Denote by $q$ the probability that a single box is empty.
We have:
\begin{equation}\label{eq:q_ub}
  q = (1 - lh^{d-1})^n\le \exp({-nlh^{d-1}}).
\end{equation}

Denote by $\AntiP(h)$ the probability that no path exists in $\channel(h)$.
We are going to prove the following lemma.
\begin{lemma}
  \label{lemma:AntiP}
  There exists constants $q_0\in (0,1)$ and $c>0$ such that if $q<q_0$ then
  we have
  \begin{equation}
    \label{eq:cut_ub}
    \AntiP(h)\le (cq)^{2^{d-3}}.
  \end{equation}
\end{lemma}
\begin{proof}
  The proof will proceed by induction on $d$. We will thus use
  the notation $\AntiP_d(h)$ and $\channel_d(h)$ to underline
  which dimension $d$ we are currently referring to.

  The crux of the proof is to prove the induction base $d=2$,
  i.e., the bound \begin{equation*}\AntiP_2(h)\leq \sqrt{cq}\end{equation*} that holds
  for all $q<q_0$ for some constants $c,q_0$.
  This bound is proved in Section~\ref{sec:antiP_2d}.

For larger $d$ it is enough to prove that the bound
  \begin{equation*}
    \AntiP_d(h)\le\bigl(\AntiP_{d-1}(h)\bigr)^2.
  \end{equation*}
  holds.
  Let $s\in\{-1, 1\}$.
  Consider a subchannel $\channel^s_d(h)$ of the channel $\channel_d(h)$
  that
  is composed of the reachable boxes $B=R(z_0, z_1, \ldots, z_{d-1})$
  fulfilling the following conditions:
  \begin{itemize}
    \item $z_0=2k$ for some integer $k$, $0\le k\le 2K$,
    \item $|z_i|\le\min(k,2K-k)$ for all $i=1,\ldots,d-2$,
    \item $z_{d-1}=s\cdot\min(k,2K-k)$,
    \item $z_i\equiv k\pmod 2$.
  \end{itemize}
  Observe that the above conditions say
  that $B$ is a reachable box in
  $\channel_d(h)$
  with additional constraint $z_{d-1}=s\cdot\min(k,2K-k)$,
  which can also be written as $z_{d-1}=s\cdot\min(z_0,4K-z_0)/2$.

  Now one can see that $\channel^s_d(h)$ has exactly the same structure
  as $\channel_{d-1}(h)$:
  we can jump between boxes $R(z_0, \ldots, z_{d-2})$ and
  $R(z'_0, \ldots, z'_{d-2})$ in channel $\channel_{d-1}(h)$
  if and only if
  we can jump between boxes \begin{equation*}R(z_0, \ldots, z_{d-2},s\cdot\min(z_0,4K-z_0)/2)\end{equation*} and
  \begin{equation*}R(z'_0, \ldots, z'_{d-2},s\cdot\min(z'_0,4K-z'_0)/2)\end{equation*} in channel $\channel^s_d(h)$.
  Therefore the probability that no path exists in $\channel^s_d(h)$
  is bounded by $\AntiP_{d-1}(h)$.

  Observe that $\channel^{-1}_d(h)$ and $\channel^1_d(h)$ share
  only the corner boxes $R(0,0,\dots,0)$ and $R(4K,0,\ldots,0)$.
  Thus if no path exists in $\channel_d(h)$, there must be no paths
  in $\channel^{-1}_d(h)$ and $\channel^1_d(h)$ independently.
  This clearly happens with probability at most
  $\bigl(\AntiP_{d-1}(h)\bigr)^2$.
\end{proof}

\section{Choosing the size of $i$-th bounding box}\label{sec:explicit}
In this section we show how we derive the bound of Lemma~\ref{lemma:barP}
from Lemma~\ref{lemma:AntiP}.
We will also be able to explicitly define the value $\imax$ and the function
$\chlenub(i)$
so that the asymptotic bounds~(\ref{eq:imax-asymp})~and~(\ref{eq:chlenub-asymp}) hold.

Suppose that for a fixed $i$ we pick such $h_i$ that $\chlenub(i)=w\sqrt{1+40^2(d-1)(h_i/r)^2}$.
Then, by Lemma~\ref{lemma:channel_len}, a path in $\channel(h_i)$ certifies the existence
of a $s\to t$ path in $G$ of length at most $\chlenub(i)$.
Such a path is clearly contained in $\ellipsoid(i)$, and thus
also in $\bbox(i)$.
As a result, we conclude
\begin{equation*}
  \bar{P}(i)\leq \AntiP(h_i).
\end{equation*}
Given this, and
since we want the probability $\bar{P}(i)$ to decay exponentially
with $i$,
we would like to choose $h_i$ in a such way that
$\AntiP(h_i)\le e^{-i}$, which will imply $\bar{P}(i)\leq e^{-i}$.

Suppose $\exp(-nlh^{d-1})<q_0$, where $q_0$ is the constant of Lemma~\ref{lemma:AntiP}.
By combining inequality~\eqref{eq:q_ub} and the bound of Lemma~\ref{lemma:AntiP}, we have
\begin{equation*}
  \AntiP(h)\le \exp\left(2^{d-3}(\log c - nlh^{d-1})\right).
\end{equation*}
In order to guarantee $\AntiP(h_i)\leq e^{-i}$, it is thus enough to have
\begin{align}
  2^{d-3}(\log c - nlh_i^{d-1})&\le-i\notag\\
  \log c + \frac{i}{2^{d-3}}&\le nlh_i^{d-1},\label{eq:step_to_h_0}
\end{align}
and
\begin{equation*}\label{eq:bound_q0}
  \log{\frac{2}{q_0}}\leq nlh^{d-1}_i.
\end{equation*}
Let $c'$ be such a positive constant that for $i\ge 1$ we have
\begin{equation}\label{eq:cprim}
  \max\left(\log{\frac{2}{q_0}},\log c + \frac{i}{2^{d-3}}\right)\le c'\cdot i.
\end{equation}
Now let $h_0$ be such that $h_0^{d-1}=\frac{c'}{nl}$, and let
\begin{equation}\label{eq:hi}
  h_i=h_0\cdot i^{\frac{1}{d-1}}.
\end{equation}
Then we have
\begin{equation*}
  \max\left(\log{\frac{2}{q_0}},\log c + \frac{i}{2^{d-3}}\right)\leq c'\cdot i=c'\cdot \left(\frac{h_i}{h_0}\right)^{d-1}=c'\cdot h_i^{d-1}\cdot \frac{nl}{c'}=nlh_i^{d-1}.
\end{equation*}
So indeed, if $h_i$ is defined as in (\ref{eq:hi}), we have
$\AntiP(h_i)\leq e^{-i}$.
So the explicit formula for $\chlenub(i)$ is:
\begin{equation*}\chlenub(i)=w\sqrt{1+40^2(d-1)\left(\frac{c'i}{nlr^{d-1}}\right)^{\frac{2}{d-1}}},\end{equation*}
  where $c'$ is a constant defined in~(\ref{eq:cprim}) and $l=\Theta(r)$ is as defined in~(\ref{eq:ldef}).
It is now verified that $\chlenub(i)$ indeed satisfies the
asymptotic formula~(\ref{eq:chlenub-asymp}) from Section~\ref{sec:oracle}.

The above proof derivation of $\bar{P}(i)\leq e^{-i}$ is only
correct if $h_i$ is not too large. Namely, recall that the bound~(\ref{eq:h_bound})
requires that
\begin{equation}
  h_i\le\frac18\sqrt{\frac7{d-1}}\cdot r.
\end{equation}
Since $h_i$ is an increasing function of $i$, this imposes a constraint
on maximum possible $i=\imax$ allowed. Hence, we need to have

\begin{align*}
  \left(\frac{c'\cdot \imax}{nl}\right)^{\frac{1}{d-1}}&\le\frac18\sqrt{\frac7{d-1}}\cdot r.\\
  \imax&= \left\lfloor \frac{1}{c'}\cdot  \left(\frac18\sqrt{\frac7{d-1}}\cdot r\right)^{d-1}\cdot nl\right\rfloor=\Theta(nr^d).
\end{align*}

Observe that the above definition of $\imax$ agrees with the bound~\eqref{eq:imax-asymp} from Section~\ref{sec:oracle}.

\section{Existence of a path in a two-dimensional grid}
\label{sec:antiP_2d}

Recall that our goal is to prove the induction
base of Lemma~\ref{lemma:AntiP} for $d=2$.
More concretely, we need to prove $\AntiP_2(h)\leq \sqrt{cq}$
for a sufficiently small $q<q_0$ and some positive constant $c$.

\paragraph{Grid formulation.} It is beneficial to reformulate our problem in terms
of reachability in directed grids.
Suppose we are given a two-dimensional grid with corners
in $(0,0)$ and $(n,n)$.\footnote{In this section we completely
forget about the graph $G$ and use $n$ to denote the grid size.}
The grid partitions $[0,n]\times [0,n]$
into $n^2$ square cells: we identify the cells
by the coordinates of its upper right corner.

\newcommand{\AntiQ}{\widetilde{P}}

The cells can be \emph{on} or \emph{off}. We consider paths from cell $(1,1)$
to cell $(n,n)$, where one can go from cell $a$
to cell $b$ if $b$ is the upper or the right neighbor of $a$
and both these cells are on.
Since one cannot go from cell $b$ to cell $a$ in this
case, the possible movements between adjacent cells
are described using a directed graph.

Each cell is \emph{off} with probability $q$ and \emph{on}
with probability $p=1-q$, independently from all the other cells.
Our goal is to upper-bound the probability $\AntiQ(q)$ that there
is no path between $(1,1)$ and $(n,n)$ using a function of $q$.

\paragraph{Correspondence to the original problem.}
Let us now describe how this reformulated problem corresponds
to the original problem.
Note that the boxes in channel $\channel(h)$ in fact
form a $n\times n$ grid, where $n=K+1$. The correspondence is as follows:
we map the box $R(x,y)$ of $\channel(h)$ to cell $(x/4-y/2+1,x/4+y/2+1)$
of the grid.
Then, there is a 1-1 correspondence between
path from $R(0,0)$ to $R(4K,0)$ as in Definition~\ref{def:path},
and paths between cell $(1,1)$ to $(n,n)$ that can only
proceed upwards or to the right.
We have $\AntiP(h)=\AntiQ(q)$.

\paragraph{Proof.}
Recall that our goal is to upper-bound the probability $\AntiQ(q)$ that there
is no path between $(1,1)$ and $(n,n)$ using a function of $q$.
To this end, consider an event when such a path does not exist.
Consider the last cell $a$ reachable from $(1,1)$
out of $(1,1),(1,2),\ldots,(1,n),(2,n),\ldots,(n-1,n)$ (i.e.,
from the ``topmost'' possible path).
Similarly, let $b$ be the last cell out
of
\begin{equation*}
  (1,1),(2,1),\ldots,(n,1),(n,2),\ldots,(n,n-1)
\end{equation*}
that is reachable from $(1,1)$ (i.e., from the ``bottommost'' possible path).

We distinguish four cases depending on the pair $a,b$:
\begin{enumerate}
  \item $a=(1,k)$ and $b=(l,1)$ for some $1\leq k,l<n$,
  \item $a=(n-k,n)$ and $b=(n,n-l)$ for some $1\leq k,l<n$,
  \item $a=(1,k)$ and $b=(n,n-l)$ for some $1\leq k,l<n$,
  \item $a=(n-k,n)$ and $b=(l,1)$ for some $1\leq k,l<n$,
\end{enumerate}

For $i=1,2,3,4$, let $\AntiQ_i(q)$ be the probability
that there is no $(1,1)\to (n,n)$ path and case $i$ occurs.
Clearly, $\AntiQ(q)=\sum_{i=1}^4 \AntiQ_i(q)$.

Consider the first case when $a=(1,k)$ and $b=(l,1)$,
where $1\leq k,l< n$.
Consider the ``contour'' of the area reachable from $(1,1)$
obtained by going around
the area's boundary while keeping the right hand in contact with it
at all times. In particular, consider the contiguous part
$\mathcal{C}$ of that contour starting at $(0,k)$
and ending at $(l,0)$. Intuitively, since the area reachable
from $(1,1)$ is connected, and by the definition of $a,b$, the curve
$\mathcal{C}$ does not intersect the grid's boundary
except at its endpoints $(0,k),(l,0)$. See Figure~\ref{fig:grid}.

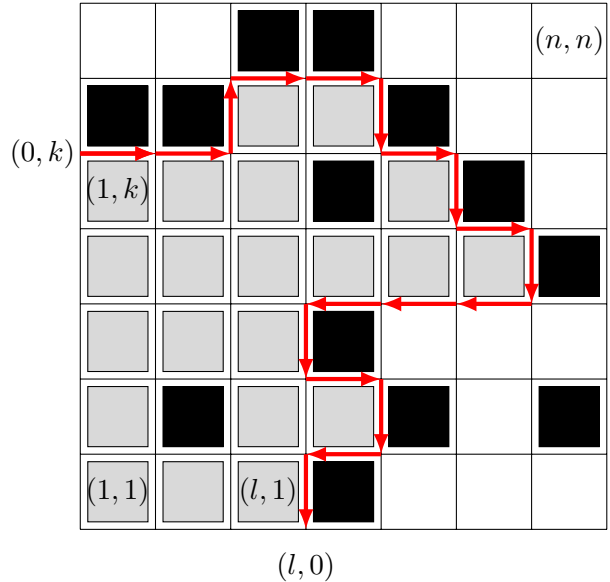
\begin{figure}[h!]
  \centering
  \begin{tikzpicture}
\draw (0,0) grid (7,7);
    \draw[fill=black] (0.1,5.1) rectangle (0.9,5.9);
    \draw[fill=black] (1.1,5.1) rectangle (1.9,5.9);

    \draw[fill=black] (2.1,6.1) rectangle (2.9,6.9);
    \draw[fill=black] (3.1,6.1) rectangle (3.9,6.9);

    \draw[fill=black] (4.1,5.1) rectangle (4.9,5.9);
    \draw[fill=black] (5.1,4.1) rectangle (5.9,4.9);
    \draw[fill=black] (6.1,3.1) rectangle (6.9,3.9);

    \draw[fill=black] (3.1,4.1) rectangle (3.9,4.9);

    \draw[fill=black] (3.1,2.1) rectangle (3.9,2.9);

    \draw[fill=black] (4.1,1.1) rectangle (4.9,1.9);

    \draw[fill=black] (1.1,1.1) rectangle (1.9,1.9);

    \draw[fill=black] (6.1,1.1) rectangle (6.9,1.9);

    \draw[fill=gray!30] (0.1,0.1) rectangle (0.9,0.9);

    \draw[fill=gray!30] (0.1,1.1) rectangle (0.9,1.9);
    \draw[fill=gray!30] (0.1,2.1) rectangle (0.9,2.9);
    \draw[fill=gray!30] (0.1,3.1) rectangle (0.9,3.9);
    \draw[fill=gray!30] (0.1,4.1) rectangle (0.9,4.9);

    \draw[fill=gray!30] (1.1,2.1) rectangle (1.9,2.9);
    \draw[fill=gray!30] (1.1,3.1) rectangle (1.9,3.9);
    \draw[fill=gray!30] (1.1,4.1) rectangle (1.9,4.9);

    \draw[fill=gray!30] (2.1,2.1) rectangle (2.9,2.9);
    \draw[fill=gray!30] (2.1,3.1) rectangle (2.9,3.9);
    \draw[fill=gray!30] (2.1,4.1) rectangle (2.9,4.9);

    \draw[fill=gray!30] (2.1,5.1) rectangle (2.9,5.9);
    \draw[fill=gray!30] (3.1,5.1) rectangle (3.9,5.9);

    \draw[fill=gray!30] (3.1,3.1) rectangle (3.9,3.9);
    \draw[fill=gray!30] (4.1,3.1) rectangle (4.9,3.9);
    \draw[fill=gray!30] (5.1,3.1) rectangle (5.9,3.9);

    \draw[fill=gray!30] (4.1,4.1) rectangle (4.9,4.9);

    \draw[fill=black] (3.1,0.1) rectangle (3.9,0.9);

    \draw[fill=gray!30] (1.1,0.1) rectangle (1.9,0.9);
    \draw[fill=gray!30] (2.1,0.1) rectangle (2.9,0.9);

    \draw[fill=gray!30] (2.1,1.1) rectangle (2.9,1.9);
    \draw[fill=gray!30] (3.1,1.1) rectangle (3.9,1.9);

    \draw[ultra thick,color=red,-latex] (0,5) to (1,5);

    \draw[ultra thick,color=red,-latex] (0,5) to (1,5);
    \draw[ultra thick,color=red,-latex] (1,5) to (2,5);
    \draw[ultra thick,color=red,-latex] (2,5) to (2,6);
    \draw[ultra thick,color=red,-latex] (2,6) to (3,6);
    \draw[ultra thick,color=red,-latex] (3,6) to (4,6);
    \draw[ultra thick,color=red,-latex] (4,6) to (4,5);
    \draw[ultra thick,color=red,-latex] (4,5) to (5,5);
    \draw[ultra thick,color=red,-latex] (5,5) to (5,4);
    \draw[ultra thick,color=red,-latex] (5,4) to (6,4);
    \draw[ultra thick,color=red,-latex] (6,4) to (6,3);
    \draw[ultra thick,color=red,-latex] (6,3) to (5,3);
    \draw[ultra thick,color=red,-latex] (5,3) to (4,3);
    \draw[ultra thick,color=red,-latex] (4,3) to (3,3);
    \draw[ultra thick,color=red,-latex] (3,3) to (3,2);
    \draw[ultra thick,color=red,-latex] (3,2) to (4,2);
    \draw[ultra thick,color=red,-latex] (4,2) to (4,1);
    \draw[ultra thick,color=red,-latex] (4,1) to (3,1);
    \draw[ultra thick,color=red,-latex] (3,1) to (3,0);

    \node at (0.5,4.5) {$(1,k)$};
    \node at (2.5,0.5) {$(l,1)$};

    \node at (-0.5,5) {$(0,k)$};
    \node at (3,-0.5) {$(l,0)$};

    \node at (0.5,0.5) {$(1,1)$};
    \node at (6.5,6.5) {$(n,n)$};
\end{tikzpicture}
\caption{The black cells are precisely those that are off. The reachable area is in gray. The non-reachable cells that are on are white.
The red arrows a part of contour~$\mathcal{C}$ from $(0,k)$ to $(l,1)$ with $s=18$ steps,
assuming we are in case 1.\label{fig:grid}}
\end{figure}

Observe that the walk around $\mathcal{C}$
consists of a number $s$ of unit-length \emph{steps}, each going
either up (U), down (D), left (L), or right (R).
Let $c_{\text{U}},c_{\text{D}},c_{\text{L}},c_{\text{R}}$
denote the counts of the respective types of steps
in $\mathcal{C}$.
Clearly, we have ${c_{\text{D}}-c_{\text{U}}=k}$
and $c_{\text{R}}-c_{\text{L}}=l$.
Hence, $c_{\text{D}}\geq c_{\text{U}}$ and $c_{\text{R}}\geq c_{\text{L}}$,
and therefore $c_{\text{D}}+c_{\text{R}}\geq \frac{1}{2}(c_{\text{D}}+c_{\text{U}}+c_{\text{L}}+c_{\text{R}})=s/2$.
We also have $s\geq k+l$.

Note that for each ``down'' step $(x,y)\to(x,y-1)$,
the cell $(x+1,y)$ is necessarily \emph{off},
since otherwise we would reach it.
Similarly, for each ``right'' step $(x,y)\to(x+1,y)$,
the cell $(x+1,y+1)$ is necessarily off.
As the steps in $\mathcal{C}$ are distinct,
each cell that is off can be ``charged'' this way
to at most two steps (at most one ``right'' step, and at most one ``down'').
As a result, $\mathcal{C}$ certifies the existence
of at least $\frac{c_{\text{D}}+c_{\text{R}}}{2}\geq s/4$
cells that are off.

Let $\AntiQ(q,\mathcal{C})$ be the probability
that an $s$-step curve $\mathcal{C}$ is the $(0,k)\to (l,0)$ part of the
contour of the reachable area.
So, we have $\AntiQ(q,\mathcal{C})\leq q^{\lceil s/4\rceil}\leq q^{s/4}.$

On the other hand, given $k,l,s$, the number
of possible $s$-step curves $\mathcal{C}$ from $(0,k)$ to $(l,0)$
is at most $3^s$, as
each subsequent step can be chosen to be
in at most $3$ distinct directions.
As a result, the probability $\AntiQ_1(q)$ that the case~1 arises,
i.e., $a$ is of the
form $(1,k)$, and $b$ is of the form $(l,1)$
for $1\leq k,l<n$ is no more than:
\begin{equation*}\AntiQ_1(q)\leq \sum_{k,l\geq 1}\sum_{s\geq k+l}3^s\cdot q^{s/4}=\sum_{k,l\geq 1}\sum_{s\geq k+l}\left(3q^{1/4}\right)^{s}.\end{equation*}
Set $\alpha=3q^{1/4}$.
Then we have:
\begin{equation*}\AntiQ_1(q)\leq \sum_{k,l\geq 1}\alpha^{k+l}\frac{1}{1-\alpha}=\frac{1}{1-\alpha}\sum_{k\geq 1}\alpha^{k+1}\cdot\frac{1}{1-\alpha}=\frac{\alpha^2}{(1-\alpha)^3}.\end{equation*}

The second case when $a=(n-k,n)$ and $b=(n,n-l)$ for $1\leq k,l<n$
is symmetric
and leads to the same bound.
Thus, $\AntiQ_2(q)\leq \frac{\alpha^2}{(1-\alpha)^3}$.

In the third case we have $a$
of the form $(1,k)$ and $b$ of the form $(n,n-l)$
for some $k,l\in \{1,\ldots,n-1\}$.
We consider (parts of) contours $\mathcal{C}$ starting
at $(0,k)$ and ending at $(n,n-l)$ and
thus $c_{\text{R}}-c_{\text{L}}=n$
and $|c_{\text{D}}-c_{\text{U}}|\leq n$.

We also have
\begin{equation*}s=c_{\text{R}}+c_{\text{L}}+c_{\text{D}}+c_{\text{U}}\leq 2c_{\text{R}}-n+c_{\text{D}}+(c_{\text{D}}+|c_{\text{D}}-c_{\text{U}}|)\leq 2(c_{\text{R}}+c_{\text{D}}),\end{equation*}
so again, by the same reasoning, a (part of) contour $\mathcal{C}$ with $s$ steps
certifies that at least $\lceil s/4\rceil$ cells are off,
and we can obtain the same bound
$\AntiQ_3(q)\leq \frac{\alpha^2}{(1-\alpha)^3}$.
on the probability that case $3$ arises.
Case $4$ is, again, symmetric to case $3$.

Since any of the described $4$ cases can apply, the
probability that one
cannot reach cell $(n,n)$ from cell $(1,1)$, is bounded by:
\begin{equation*}\AntiQ(q)=\sum_{i=1}^4 \AntiQ_i(q)\leq \frac{4\alpha^2}{(1-\alpha)^3}.\end{equation*}
Assume $q<\frac{1}{2^{10}\cdot 3^4}$. Then $\alpha=\frac{1}{2^{2.5}}<1/2$, and thus:

\begin{equation*}\AntiQ(q)\leq \frac{4\alpha^2}{(1-\alpha)^3}\leq \frac{4\alpha^2}{(1/2)^3}\leq  32\alpha^2=32\cdot 9\cdot q^{1/2}<1.\end{equation*}

To conclude, we have proved that for $c=(32\cdot 9)^2$ and $q<\frac{1}{2^{10}\cdot 3^4}$
we indeed have $\AntiQ(q)\leq \sqrt{cq}$ as desired.

\subsection{A shortcoming in~\cite{sedgewick1986}.}
\label{ashortcoming}
Sedgewick and Vitter~\cite[pages 41-42]{sedgewick1986} also derive a bound $\AntiQ(q)=O(\poly{q})$.
However, we believe they argument to be flawed.
When bounding $\AntiQ(q)$, they argue that unless a (directed, as defined above)
path from $(1,1)$ to $(n,n)$ exists there has to be an ``antipath''.
An antipath is defined to be a sequence of cells that are ``off''
such that each subsequent cell is a neighbor of the previous
one, and $(1,1)$ cannot reach $(n,n)$.
However, as the example in Figure~\ref{fig:cexp} shows, $(n,n)$ may become
unreachable from $(1,1)$ even if no antipath exists,
i.e., when the cells that are off do not form a path, regardless
of how exactly we define neighborhood between cells (e.g., neighboring
sides, or neighboring corners).
As a result, their bound on $\AntiQ(q)$ does not cover
all possible cases and thus underestimates the probability that
no path from $(1,1)$ to $(n,n)$ exists.

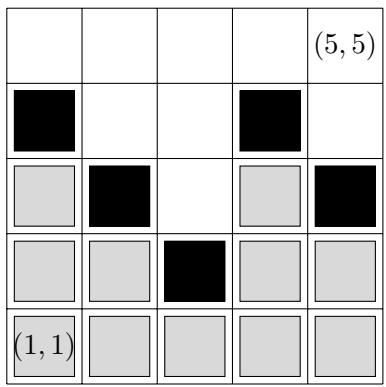
\begin{figure}[h!]
  \centering
  \begin{tikzpicture}
\draw (0,0) grid (5,5);

    \draw[fill=black] (0.1,3.1) rectangle (0.9,3.9);
    \draw[fill=black] (1.1,2.1) rectangle (1.9,2.9);
    \draw[fill=black] (2.1,1.1) rectangle (2.9,1.9);

    \draw[fill=black] (3.1,3.1) rectangle (3.9,3.9);
    \draw[fill=black] (4.1,2.1) rectangle (4.9,2.9);

    \draw[fill=gray!30] (0.1,2.1) rectangle (0.9,2.9);
    \draw[fill=gray!30] (0.1,1.1) rectangle (0.9,1.9);
    \draw[fill=gray!30] (0.1,0.1) rectangle (0.9,0.9);

    \draw[fill=gray!30] (1.1,1.1) rectangle (1.9,1.9);
    \draw[fill=gray!30] (1.1,0.1) rectangle (1.9,0.9);

    \draw[fill=gray!30] (2.1,0.1) rectangle (2.9,0.9);
    \draw[fill=gray!30] (3.1,0.1) rectangle (3.9,0.9);
    \draw[fill=gray!30] (4.1,0.1) rectangle (4.9,0.9);

    \draw[fill=gray!30] (3.1,1.1) rectangle (3.9,1.9);
    \draw[fill=gray!30] (4.1,1.1) rectangle (4.9,1.9);

    \draw[fill=gray!30] (3.1,2.1) rectangle (3.9,2.9);

    \node at (0.5,0.5) {$(1,1)$};
    \node at (4.5,4.5) {$(5,5)$};
  \end{tikzpicture}
  \caption{The cells that are \emph{off} (black) do not have to form a path
  in any sense to disconnect $(1,1)$ from $(5,5)$.
  The area reachable from $(1,1)$ is gray.
  \label{fig:cexp}}
\end{figure}

\bibliographystyle{plainurl}
\bibliography{runit}

\end{document}